\documentclass{article}
\usepackage{amsmath}
\usepackage{amsfonts}
\usepackage{amssymb}
\usepackage{latexsym}
\usepackage{amsthm}

\newtheorem{theorem}{Theorem}
\newtheorem{proposition}[theorem]{Proposition}

\newtheorem{definition}[theorem]{Definition}

\newcommand{\A} {\mathbf{A}}

\providecommand{\keywords}[1]
{
	{\noindent
	\textbf{\textit{Keywords:}} #1}
}

\title{On the Advice Complexity of Online Unit Clustering}
\author{Judit Nagy-Gy\"orgy \thanks {Bolyai Institute, University of
Szeged, Aradi V\'ertan\'uk tere 1, H-6720 Szeged, Hungary, email:
Nagy-Gyorgy@math.u-szeged.hu}}
\date{}

\begin{document}
\frenchspacing
\maketitle

\begin{abstract}
	In online unit clustering, points of a metric space arriving one by one must be partitioned into clusters of diameter at most 1, where the cost is the number of clusters.	
	This paper gives linear upper and lower bounds on the advice complexity of 1-competitive online unit clustering algorithms, in terms of the number of points in $\mathbb{R}^d$ and $\mathbb{Z}^d$. 
\end{abstract}

\vspace{10pt}
\vspace{15pt}

\section{Introduction}

Clustering problems are fundamental and arise in a wide variety of applications. In clustering problems, the goal is to find a partition of a set of points in a metric space for which an objective function is optimized. The objective functions can be various. The elements of such partitions are called clusters. 

In this paper, we consider problems where each cluster must be coverable by a unit ball, trying to minimize the number of subsets used. 
In the online unit clustering model, points are presented to the algorithm one by one and must be assigned to clusters upon arrival, and this assignment cannot be changed later. This problem was  introduced by Chan and Zarrabi-Zadeh in \cite{CZZ}. 
Online unit covering, previously studied by Charikar et al. in \cite{CCFM},
differs in that the exact location of the ball covering a cluster must be fixed when the cluster is opened. Unit covering is the offline version of both problems.
The space is $\mathbb{R}^d$, $d\in\mathbb{N}$ with norm $L_\infty$, so balls are intervals, squares, or cubes. 

We evaluate the efficiency of online algorithms in terms of the competitive ratio (see \cite{BOR}), where the online algorithm is compared to the optimal offline algorithm. An online algorithm is $c$-competitive if its cost is at most $c$ times the optimal cost. 
For randomized algorithms, the expected value of the cost is compared to the optimum.

For online unit covering in $\mathbb{R}^d$, Charikar et al. gave a deterministic algorithm of competitive ratio $O(2^d d \log d)$ and tight bounds of 2 for $d=1$ and 4 for $d = 2$, later Dumitrescu and T\'oth in \cite{DT} gave a lower bound of $d+1$ on the competitive ratio of any deterministic online algorithm. The authors of \cite{DT} also designed an $O(d^2)$-competitive randomized online algorithm for online unit covering in $\mathbb{Z}^d$.

The competitive ratio of deterministic and random algorithms for online unit clustering has been investigated in several papers (see \cite{CZZ, CCFM, CsEIL, DI, DT, EL, ELS, ES, KK, ZZC}).
Table~\ref{table} contains a summary of the best known bounds on competitive ratios in $\mathbb{R}^d$. 

\begin{table}[h!]
	\centering
\begin{tabular}{l | l l l}
	  & $d=1$ & $d=2$ & $d>2$\\
	 \hline
	deterministic upper bounds& 5/3\ \ \cite{EL} & 10/3\ \ \cite{EL} & $2^d\cdot 5/6$ \cite{EL} \\
	deterministic lower bounds& 13/8\ \ \cite{KK}\ \ \ & 13/6\ \ \cite{EL}\ \ \ & $\Omega(d)$ \cite{DT}\\
	randomized lower bounds& 3/2\ \ \cite{ES} & 11/6\ \ \cite{ES} & $\Omega(d)$ \cite{DT}\\
	\hline
\end{tabular}
\caption{A summary of the best known bounds on the competitive ratio for unit clustering in $\mathbb{R}^d$}
\label{table}
\end{table}

\noindent It should be mentioned that Dumitrescu and T\'oth showed in \cite{DT} that the competitive ratio of the greedy algorithm is at most $2^{d-1}+\frac{1}{2}$ for online unit clustering in $\mathbb{Z}^d$.

The notion of advice complexity for online algorithms was introduced by the authors of \cite{DKP}. The main question is: How many bits of advice are necessary and sufficient to achieve a competitive ratio of $c$? This involves determining the number of bits for optimality. The results of the following two sections are obtained in the \emph{tape model} introduced in \cite{BKKKM}. In this model, the online algorithm can read an infinite advice tape written by the oracle, and the advice complexity is simply the number of bits read.
For more information on advice complexity, see the survey paper \cite{BFKLM}.
Mikkelsen showed in \cite{Mik} that an algorithm with $o(n)$ bits of advice for unit clustering in $\mathbb{R}$ must be at least 3/2-competitive.

As Epstein and van Stee noted in \cite{ES}, online clustering is an online graph coloring problem. If we think of the clusters as colors and the points as vertices, then an edge occurs between two points if they are too far apart to be colored with the same color. This implies that upper bounds on the advice complexity of online graph coloring hold for online unit clustering as well.

In section~\ref{lower}, lower bounds $d(n-2^{d+1}+1)$ and $\frac{d}{1+2d}$ on the number of advice bits necessary in $\mathbb{R}^d$ and $\mathbb{Z}^d$, respectively, to achieve a competitive ratio of 1 are presented. Note that the latter lower bound is valid for both online models in $\mathbb{Z}^d$, but the theorem is formulated for online unit clustering only.

Section~\ref{upper} contains three algorithms with competitive ratio 1 for unit clustering reading $nd$, $\lfloor n/2\rfloor$, and $\lfloor (d+1)n/2\rfloor$ bits of advice in $\mathbb{R}^d$, $\mathbb{Z}$, and $\mathbb{Z}^d$, respectively. The last upper bound holds for both online models in $\mathbb{Z}^d$.

\section{Lower bounds}\label{lower}

\begin{definition}
	An online algorithm $\A$ needs different advice words for inputs $I_1$ and $I_2$ to serve them optimally if whenever $\A$ serves $I_1$ optimally reading $\underline{w}_1$ and $\A$ serves $I_2$ optimally reading $\underline{w}_2$, none of $\underline{w}_1$ and $\underline{w}_2$ is a prefix of the other (in particular, they are not the same and neither is empty). 
\end{definition}

\begin{proposition}\label{differ}
	If for input sequences $I_1$ and $I_2$ there exists $k>1$ such that $I_1^{(k)}=I_2^{(k)}=I'$, where $I_j^{(k)}$ denotes the subsequence consisting of the first $k$ elements of $I_j$, furthermore $I'$ must have different clusterings in the optimal service of $I_1$ and $I_2$, then $\A$ needs different advice words for inputs $I_1$ and $I_2$ to serve them optimally.
\end{proposition}

\begin{proof}
	Suppose, to the contrary, that online algorithm $\A$ does not need different advice words for inputs $I_1$ and $I_2$. Suppose that $\A$ reads $\underline{w}_1$ while serving $I'$ if the entire input is $I_1$, and it reads $\underline{w}_2$ while serving $I'$ if the entire input is $I_2$. 
	We can assume without loss of generality that $|\underline{w}_1|\le|\underline{w}_2|$. If $\underline{w}_1$ is a prefix of $\underline{w}_2$, then the behavior of the algorithm while serving $I'$ and reading $\underline{w}_1$ is the same in both cases, so $\A$ cannot serve both $I_1$ and $I_2$ optimally.
\end{proof}

\begin{theorem}
There is no algorithm for online unit clustering in $\mathbb{R}^d$ that serves all request sequences of length at least $n$ optimally reading less than $d(n-2^{d+1}+1)$ bits of advice.
\end{theorem}

\begin{proof}
	At first, consider $d=1$. Fix $m\ge 0$, let $n=m+3$ and set $J=J_m=\{-1,1\}^{m}$. Construct an input sequence $\mathbf{r}=\mathbf{r_j}$ for each $\mathbf{j}=(j_1,\ldots,j_{m})\in J$ in the following way:
	\begin{eqnarray*}
		r_1 &=& 1,\\
		r_2 &=& 1/2,\\
		r_{i+2} &=& r_{i+1}+j_{i}/2^{i+1},\quad i=1,\ldots, m-1\\
		r_{m+2} &=& r_{m+1}-j_{m},\\
		r_{m+3} &=& r_{m+1}+j_{m}(1+1/2^{m}).
	\end{eqnarray*}
	Note that one of $r_{m+2}$ and $r_{n+3}$ is negative and the other is between 1 and 2, and their distance is $$|r_{m+3}-r_{m+2}|=|(2+1/2^{m})j_m|=2+1/2^{m}.$$
	The optimal solution is unique and has two clusters for each request sequence $\mathbf{r_j}$. If $j_{m}=1$ then the clusters covered by intervals of length 1 are 
	$$[r_{n-1}, r_{n-2}] \quad \textrm{and}\quad [r_{n-2}+1/2^{n-3}, r_{n}],$$ 
	otherwise the clusters covered by intervals of length 1 are
	$$[r_{n}, r_{n-2}-1/2^{n-3}]\quad \textrm{and}\quad [r_{n-2}, r_{n-1}],$$
	since the distance between any two requests is at least $1/2^{n-3}$, and $r_2,\ldots, r_{n-2}\in (0,1)$. Note that the first interval contains 0 and the second interval contains 1.
	
	Note that if $j_i=1$ then $r_{i+1}$ is in the cluster covered by the first interval because 
	if $r_{i+2}>r_{i+1}$ for some $1\le i <n-4$ then $r_{i'}-1/2^{n-3}\ge r_{i+1}$ for all $i+2\le i'\le n-2$ and by definition of $r_{n-1}$ and $r_n$, moreover, 
	if $j_i=-1$, then $r_{i+1}$ is in the cluster covered by the second interval because 
	if $r_{i+2}<r_{i+1}$ for some $1\le i <n-4$, then $r_{i'}+1/2^{n-3}\le r_{i+1}$ for all $i+2\le i'\le n-2$ and by definition of $r_{n-1}$ and $r_n$.

	There are $2^{n-3}$ input sequences, so if we prove that any algorithm needs different advice words for each pair of inputs to produce an optimal clustering, then we are done.

	Consider request sequences $\mathbf{r}_{\mathbf{j}}$ and $\mathbf{r}_{\mathbf{j}'}$. Suppose that $j_i=-1$, $j'_i=1$ and $j_{k}=j'_{k}$ for all $k<i$. Therefore, the first $i+1$ requests of the two inputs are identical. Furthermore, the $(i+1)$th request of $\mathbf{r}_{\mathbf{j}}$ and 1 are in the same cluster of the optimal clustering, but the $(i+1)$th request of $\mathbf{r}_{\mathbf{j}'}$ and 1 are in different clusters of the optimal clustering. Thus, any algorithm needs different advice words to distinguish the two input sequences before the $(i+2)$th request.
	
	\vspace{10pt}
	
	Now set $d>1$, $m>0$ integers and $n=m+2^{d+1}-1$. 
	For each $(\mathbf{j}_1,\ldots,\mathbf{j}_m)\in J_d^m$, define an input sequence in the following way: 
	\begin{eqnarray*}
	\{\mathbf{r}_1,\ldots,\mathbf{r}_{2^d-1}\} &=& \{0,1\}^d\setminus\{\mathbf{0}\},\\
	\mathbf{r}_{2^d} &=& \{1/2\}^d,\\ 
	\mathbf{r}_{2^d+i} &=& \mathbf{r}_{2^d+i-1}+\mathbf{j}_i/2^{i+1},\quad i=1,\ldots, m-1\\
	\{\mathbf{r}_{2^d+m},\ldots, \mathbf{r}_{2^d+m+2^d-1}\}&=& \mathbf{r}_{2^d+m-1}+\prod_{i=1}^d\{-j_{m,i},j_{m,i}(1+1/2^{m})\}
	\end{eqnarray*}
	where the last product is the $d$-ary Cartesian product. The first $2^d-1$ requests can be revealed in an arbitrary fixed order, and similarly, the last $2^d$ requests may arrive in an arbitrary order.
	
	Observations:
	\begin{itemize}
		\item If we fix a coordinate $\ell$ and consider the $\ell$th coordinate of the requests of a request sequence, deleting the zeros and the identical requests, then we essentially get a request sequence from the one-dimensional case. Note that all request sequences can be obtained in this way.
		\item $\mathbf{r}_{2^d+i}\in (0,1)^d$ for all $i=0,\ldots,m-1$.
		\item Every coordinate of $\mathbf{r}_{2^d+m+i}$ is in $(-1,0)\cup(1,2)$ for all $i=0,\ldots,2d-1$ and there is exactly one of them with all negative coordinates, denote it by $\mathbf{r}^*$.
		\item The distance between any two requests among the last $2^d$ is exactly $2+1/2^m$, therefore any two of them have to be in different clusters.
		\item The distance between any two requests is at least $1/2^m$.
	\end{itemize}
	The only optimal solution has $2^d$ clusters, they can be covered by unit balls
	$$\mathbf{r}^*+(1+1/2^{n-3})\mathbf{v}+[0,1]^d, \quad \mathbf{v}\in\{0,1\}^d.$$
	
	Note that any two of the first $2^{d-1}$ requests are in different clusters of the optimal clustering, and none of them is in the one covered by $\mathbf{r}^*+[0,1]^d$. Consider the one-dimensional request sequences consisting of the $\ell$th coordinates of $\mathbf{r}$. Then the optimal clustering of them consists of two clusters that can be covered by intervals $[r_\ell^*,r_\ell^*+1]$ and $[r_\ell^*+1+1/2^m, r_\ell^*+2+1/2^m]$.

	There are $2^{dm}$ input sequences where $m=n-2^{d+1}+1$, so if we prove that any algorithm needs different advice words for each pair of inputs to produce an optimal clustering, then we are done.
	
	Consider request sequences $\mathbf{r}_{\mathbf{j}}$ and $\mathbf{r}_{\mathbf{j}'}$. 
	Suppose that $\mathbf{j}_{k}=\mathbf{j}'_{k}$ for all $k<i$ and $\mathbf{j}_{i}\ne\mathbf{j}'_{i}$.
	This means that the first $2^d+i-1$ requests of $\mathbf{r}_{\mathbf{j}}$ and $\mathbf{r}_{\mathbf{j}'}$ are identical and the $(2^d+i)$th requests are different.
	We can assume that the $\ell$th coordinate of ${\mathbf{j}}_i$ is $-1$ and the $\ell$th coordinate of ${\mathbf{j}'_i}$ is 1. Consider the one-dimensional request sequences consisting of the $\ell$th coordinates of $\mathbf{r}_{\mathbf{j}}$ and $\mathbf{r}_{\mathbf{j}'}$. As we saw in case $d=1$, the clusters of the first $2^d+i-1$ requests of these sequences are not the same, therefore the clusters of the first $2^d+i-1$ requests of $\mathbf{r}_{\mathbf{j}}$ and $\mathbf{r}_{\mathbf{j}'}$ in the optimal clusterings cannot be the same, therefore any algorithm needs different advice words to distinguish the two input sequences before the $(2^d+i)$th request.
	\end{proof}
	Note that there is no algorithm for online unit covering in $\mathbb{R}^d$ that serves all request sequences of length $n\ge 3$ optimally reading a finite number of advice bits. Anyone can check this by looking at the input set $\{(0, r, r-1): r\in(0,1)\}$. 

	\vspace{10pt}
	\noindent\textbf{Remark} Note that using Mikkelsen's technique described in \cite{Mik} on lower bounds presented by Epstein and van Stee in \cite{ES} and by Dumitrescu and T\'oth in \cite{DT}, one can prove that 
	an algorithm with $o(n)$ bits of advice for unit clustering in $\mathbb{R}^2$ must be at least 11/6-competitive, and
	an algorithm with $o(n)$ bits of advice for unit clustering in $\mathbb{R}^d$ must be $\Omega(d)$-competitive.

\begin{theorem}\label{intlower}
	There is no algorithm for online unit clustering in $\mathbb{Z}^d$ that serves all request sequences of length $n\ge n_0$ for some $n_0$ optimally reading at most $\frac{d}{1+2d}\cdot n$ bits of advice.
\end{theorem}

\begin{proof}
	Set $m>0$ integer. First of all, consider case $d=1$.
	It is easy to see that any algorithm needs different advice words for input sequences  
	$I_0=(5, 6)$ and $I_1=(5, 6, 4, 7)$,
	because its behavior on the second request point must be different. 
	
	Suppose that $d\ge 1$ integer.
	For each $\underline{j}=(j_1,\ldots,j_{md})\in\{0,1\}^{md}$ we define a request sequence $I_{\underline{j}}$. Denote the $i$th unit vector of $\mathbb{Z}^d$ by $\mathbf{e}_{i}$. 
	
	The first part of each sequence is $I'$, which consists of the elements of $$ \{5i\mathbf{e}_1,\ \  5i\mathbf{e}_1+\mathbf{e}_\ell :1\le i\le m,\ \  1\le\ell\le d\}$$
	in an arbitrary fixed order.
	
	The second part of $I_{\underline{j}}$ consists of the elements of 
	$$\{5i\mathbf{e}_1-\mathbf{e}_\ell,\ \ 5i\mathbf{e}_1+2\mathbf{e}_\ell: j_k=1,\ k=(i-1)d+\ell,\ 1\le i\le m,\ 1\le\ell\le d\} $$ 
	in an arbitrary order.
	Now consider the case $m=1$. 
	If we consider a projection of a request sequence to an arbitrary dimension, then the image is essentially $I_0$ or $I_1$. 
	Therefore, if algorithm $\A$ serves all of them optimally, then it needs different advice words for each pair of them by Proposition~\ref{differ}. 

	If $m>1$, then the point set of an input can be partitioned into blocks such that the distance between any two blocks is greater than 1, and the request set of each block is essentially the set of points of one of the request sequences of case $m=1$. Therefore, if algorithm $\A$ serves all of them optimally, then it needs different advice words for each pair of them by Proposition~\ref{differ}.  
			
	There are $2^{md}$ request sequences such that any algorithm needs different advice words for any two of them. Recall that none of these words can be a prefix of another.
	Each request sequence consists of at least $(1+d)m$ and at most $(1+3d)m$ request points.
	Let $n_0\ge 0$ be an arbitrary constant. Suppose, to the contrary, that algorithm $\A$ serves all request sequences of length $n\ge n_0$ optimally reading at most $\frac{d}{1+2d}\cdot n$ bits of advice. Set $m\ge\lceil\frac{n_0}{d+1}\rceil$ such that $1+2d \mid m$. Denote the advice word read by $\A$ to serve $I_{\underline{j}}$ optimally by $\underline{w}_{\underline{j}}$. Let
	$$W=\left\{\underline{w}:|\underline{w}|=\frac{d(1+3d)m}{1+2d},\ \underline{w} \textrm{ is a prefix of }\underline{w}_{\underline{j}},\ \underline{j}\in\{0,1\}^{md}\right\}.$$
	The number of request sequences containing $(1+3d)m-2k$ points is $\binom{md}{k}$, moreover, if $I_{\underline{j}}$ consists of $(1+3d)m-2k$ request points, then $$|w_{\underline{j}}|\le \frac{d}{1+2d}((1+3d)m-2k),$$
	so $w_{\underline{j}}$ is a prefix of at least $2^{\frac{2kd}{1+2d}}$ elements of $W$, and if $\underline{j}\ne\underline{j}'\in\{0,1\}^{md}$ then $w_{\underline{j}'}$ is not a prefix of any of them.
	Therefore $$|W|\ge\sum_{k=0}^{md}\binom{md}{k}2^{\frac{2kd}{1+2d}}=\left(2^{\frac{2d}{1+2d}}+1\right)^{md}>2^{md\cdot\frac{1+3d}{1+2d}},$$
	but this is a contradiction.
\end{proof}
Note that this result is not sharp, one can get a higher bound on the advice complexity by more careful calculation.

	\vspace{10pt}
\noindent\textbf{Remark} Chan and Zarrabi-Zadeh in \cite{CZZ} gave a lower bound of 4/3 on the competitive ratio of any randomized algorithm for unit clustering in $\mathbb{R}$. By slightly modifying their construction (e.g. using input sequences $(5,6)$ and $(5,6,4,7)$), it is easy to see that this bound also holds in $\mathbb{Z}$. 
Thus, using Mikkelsen's technique (see in \cite{Mik}) on this lower bound, one can prove that an algorithm with $o(n)$ bits of advice for unit clustering in $\mathbb{Z}$ must be at least 4/3-competitive.

\vspace{10pt}

\section{Upper bounds}\label{upper}

We have an almost tight upper bound in $\mathbb{R}$.

\begin{theorem}\label{realupper}
	There is an online algorithm $\A_{\mathbb{R}}$ for unit clustering in $\mathbb{R}$ that reads at most $n$ bits of advice to cluster each request sequence consisting of $n$ points optimally.
\end{theorem}

\begin{proof}
	Chan and Zarrabi-Zadeh \cite{CZZ} noted that it is trivial to find an optimal solution for a given input offline, i.e. starting from the left, repeatedly define a cluster of length 1 with the leftmost unserved point as its left endpoint. Call this algorithm OPT. Observe that the clusters constructed by OPT do not intersect, and that each contains one or two integers. Let the maximum integer in cluster $C$ be the representation point of $C$. The representation point of a cluster is not necessarily a request. Clearly, different clusters have different representation points. Note that the representation point of the cluster of a request $r$ is $\lfloor r\rfloor$ or $\lfloor r\rfloor+1$. 

	\noindent Now we can describe algorithm $\A_{\mathbb{R}}$:\\  
	Assign a cluster $C_z$ to all $z\in\mathbb{Z}$ and open it (for cost 1) when the first request is assigned to it by $\A_{\mathbb{R}}$.\\
	When request point $r$ arrives, then $\A_{\mathbb{R}}$ reads an advice bit $b\in\{0,1\}$ if the representation point of the cluster created by OPT containing $r$ is $\lfloor r\rfloor+b$. 
	Assign $r$ to $C_{\lfloor r\rfloor +b}$. 
	
	The clustering produced by $\A_{\mathbb{R}}$ is optimal because it assigns a request $r$ to $C_z$ if and only if the representation point of the cluster to which $r$ is assigned by OPT is $z$.
\end{proof}

\begin{theorem}\label{realdupper}
	There is an online algorithm $\A_{\mathbb{R}^d}$ for unit clustering in $\mathbb{R}^d$ which reads at most $2dn$ bits of advice to cluster optimally each request sequence consisting of $n$ points. 
\end{theorem}

\begin{proof}
	Consider the optimal clustering algorithm OPT. We can assume that the clusters created by OPT are unit balls, i.e. a cluster is $$C=[q_1,q_1+1]\times\ldots\times[q_d,q_d+1]$$ for some $(q_1,\ldots,q_d)\in\mathbb{R}^d$. Each cluster created by OPT has a (not necessarily requested) point with integer coordinates in it. Choose the one with maximal coordinates to be the representation point of the cluster. 

	Point $\mathbf{z}\in\mathbb{Z}^d$ can be the representation point of at most $2^d$ clusters constructed by OPT. Suppose, to the contrary, that $\mathbf{z}$ is in more than $2^d$ clusters. Then all these clusters are subsets of $\mathbf{z}+[-1,1]^d$, but it can be covered by $2^d$ clusters, and this is a contradiction.

	For each $\mathbf{z}\in\mathbb{Z}^d$, if it is a representation point of a cluster constructed by OPT, then consider an order of the set $\mathcal{C}_\mathbf{z}=\{C: \mathbf{z} \textrm{ is the reference point of } C\}$ and let label $l_C$ be the sequence number of $C$ in $\mathcal{C}_\mathbf{z}$ according to this order. Each cluster can be identified by its representation point and label in this way. 
	
	Now we can describe $\A_{\mathbb{R}^d}$:\\
	When request $\mathbf{r} = (r_1,\ldots,r_d)$ arrives $\A_d$ reads bits $b_1,\ldots,b_d$ and bits $b_1',\ldots,b_d'$. Advice bit $b_i$ is 1 if and only if $z_i> \lfloor r_i\rfloor$ where $\mathbf{z}=(z_1,\ldots,z_d)$ is the representation point of the cluster $C$ created by OPT containing $\mathbf{r}$ and the word $b_1'\ldots b_d'$ encodes the label $l_C$. $\A^{(d)}$ assigns $\mathbf{r}$ to cluster $C$ identified by representation point $(\lfloor r_1\rfloor +b_1,\ldots,\lfloor r_d\rfloor +b_d)$ and label $l_C$.
	
	Algorithm $\A_{\mathbb{R}^d}$ assigns each request to the same cluster as OPT, therefore, it produces an optimal clustering.
\end{proof}

Better results can be achieved in $\mathbb{Z}^d$. At first, $d=1$ is considered.

\begin{theorem}\label{intupper}
	There is an online algorithm $\A_{\mathbb{Z}}$ for unit clustering in $\mathbb{Z}$ that reads at most $\lfloor n/2\rfloor$ bits of advice to cluster each request sequence consisting of $n$ points optimally. 
\end{theorem}

\begin{proof}
Consider optimal algorithm OPT described in the proof of Theorem~\ref{realupper}. First, let's make some observations. Obviously, each cluster contains one or two different request points in the optimal clustering. If the cluster of $p$ created by OPT does not contain another request point then $p+1$ is not in the request sequence. Therefore if $p$ is the minimal requested point in a cluster then that cluster is $[p,p+1]$, and these intervals do not intersect.

A step of algorithm $\A_{\mathbb{Z}}$ to serve request $p$ is the following. Note that each requested point is unmarked when it arrives. 
\begin{itemize}
	\item If $p$ is in an interval that covers a cluster, then assign $p$ to that cluster;
	\item otherwise, if $p-1$ is in an interval that covers a cluster and $p+1$ is a requested point, then assign request $p$ to the cluster of $p+1$, and cover that cluster with interval $[p, p+1]$;
	\item otherwise, if $p-1$ is in an interval covering a cluster, $p+1$ is not a requested point, then open a new cluster, assign request $p$ to it, and cover that cluster with interval $[p, p+1]$;
	\item otherwise, if $p+1$ is in an interval covering a cluster and $p-1$ is a requested point, then assign request $p$ to cluster of $p-1$, and cover that cluster with interval $[p-1, p]$;
	\item otherwise, if $p+1$ is in an interval covering a cluster and $p-1$ is not a requested point, then open a new cluster, assign request $p$ to it, and cover that cluster with interval $[p-1, p]$;
	\item otherwise, if both $p-1$ and $p+1$ are requested points, then mark $p-1$, $p$ and $p+1$, and read a bit of advice,
	\begin{itemize}
		\item if it is 0, then assign $p$ to the cluster of $p+1$, cover this cluster with interval $[p,p+1]$ and cover the cluster of $p-1$ with interval $[p-2,p-1]$,
		\item if it is 1, then assign $p$ to the cluster of $p-1$, cover this cluster with interval $[p-1,p]$ and cover the cluster of $p+1$ with interval $[p+1, p+2]$;
	\end{itemize}
	\item otherwise, if $p-1$ is a requested point, then mark $p-1$ and $p$, and read a bit of advice,
	\begin{itemize}
		\item if it is 1, then assign $p$ to cluster of $p-1$, cover that cluster with interval $[p-1,p]$;
		\item if it is 0, then open a new cluster, assign $p$ to it and cover this cluster with interval $[p,p+1]$, and cover the cluster of $p-1$ with interval $[p-2,p-1]$;
	\end{itemize}
	\item otherwise, if $p+1$ is a requested point, then mark $p$ and $p+1$, and read a bit of advice,
	\begin{itemize}
		\item if it is 0, then assign $p$ to cluster of $p+1$, cover that cluster with interval $[p,p+1]$;
		\item if it is 1, then open a new cluster, assign $p$ to it, and cover this cluster with interval $[p-1,p]$, and cover the cluster of $p+1$ with interval $[p+1,p+2]$;
	\end{itemize}
	\item otherwise open a new cluster and assign $p$ to it.
\end{itemize}

Note that if a cluster created by $\A_{\mathbb{Z}}$ is not covered by an interval, then it has one requested point, moreover, neither $p-1$ nor $p+1$ is requested.
This fact and the observations above imply that algorithm $\A_{\mathbb{Z}}$ finds the clustering of OPT because the algorithm covers a cluster with an interval when it reads a bit of advice that says a cluster of OPT or there is an interval covering a closely adjacent cluster.

Moreover, when $\A_{\mathbb{Z}}$ reads a bit of advice (which is 1 if and only if $[p-1,p]$ is a cluster of OPT), it marks at least two requested points that are not in an interval covering a cluster, then covers their clusters with intervals, so it reads at most $\lfloor n/2\rfloor$ bits of advice.
\end{proof}

The following algorithm needs more information to identify the first request (i.e. when to open a new cluster) and the location of each cluster (i.e. to recognize subsequent points in the cluster).

\begin{theorem}\label{gridupper}
	Fix $d>0$ integer.
	There is an online algorithm $\A_{\mathbb{Z}^d}$ for unit covering in $\mathbb{Z}^d$ that reads at most $\lfloor (d+1)n/2\rfloor$ bits of advice to cluster each request sequence consisting of $n$ points optimally. 
\end{theorem}

\begin{proof}
	A step of algorithm $\A_{\mathbb{Z}^d}$ to serve request $\mathbf{p}=(p_1,\ldots,p_d)$ is the following. 
	\begin{itemize}
		\item If $\mathbf{p}$ is in a ball that covers a cluster, then assign $\mathbf{p}$ to that cluster;
		\item otherwise open a new cluster, assign $\mathbf{p}$ to it, and read an advice bit. If the advice bit is 0, then read more $d$ bits of advice: $b_1,\ldots, b_d$ and cover the new cluster with ball $\mathbf{p}+[-b_1, 1-b_1]\times\ldots\times[-b_d,1-b_d]$. Otherwise cover the new cluster with ball $\mathbf{p}+[0,1]^d$.
	\end{itemize}
	The meaning of the first advice bit in a step: it is 1 if and only if the requested point is the only point in its cluster (and any ball containing $\mathbf{p}$ is suitable). It is easy to see that the advice word can encode an optimal clustering. Let $c_1$ be the number of clusters containing one requested point and $c_0$ the number of clusters containing at least two requested points. The number of bits read by $\A_{\mathbb{Z}^d}$ is 
	$$c_1 +(d+1)c_0\le \frac{d+1}{2}\cdot (c_1 + 2c_0)\le \frac{d+1}{2}\cdot n.$$
\end{proof}

\section{Acknowledgements}

The author was supported by J\'anos Bolyai Research Scholarship of the Hungarian Academy of Sciences.

This study was supported by the project TKP2021-NVA-09. Project no TKP2021-NVA-09 has been implemented with the support provided  by the Ministry of Culture and Innovation of Hungary from the National Research, Development and Innovation Fund, financed under the TKP2021-NVA funding scheme.


\begin{thebibliography}{99}
	\bibitem{BOR}{A. Borodin, R. El-Yaniv, Online Computation and Competitive Analysis, Cambridge University Press, 1998}
	\bibitem{BFKLM}{J. Boyar, L. M. Favrholdt, C. Kudahl, K. S. Larsen, J. W. Mikkelsen, Online Algorithms with Advice: a Survey, ACM Computing Surveys, 50(2), Article No. 19, 2017}
	\bibitem{BKKKM}{H-J. B\"ockenhauer, D. Komm, R. Kr\'alovi\v{c}, R. Kr\'alovi\v{c}, T. M\"omke, On the Advice Complexity of Online Problems, ICALP (1), LNCS 6755, pp 207--218, 2011}
	\bibitem{CZZ}{T. M. Chan, H. Zarrabi-Zadeh, A randomized algorithm for online unit clustering, Theory Comput. Syst. 45(3), 486--496, 2009}
	\bibitem{CCFM}{M. Charicar, C. Chekuri, T. Feder, R. Motwani,  Incremental Clustering and Dynamic Information Retrieval. SIAM J. Comput. 33(6), pp 1417--1440, 2004}	
	\bibitem{CsEIL}{J. Csirik, L. Epstein, C. Imreh, A. Levin, Online clustering with variable sized clusters, Algorithmica 65(2), pp 251--274, 2013}
	\bibitem{DI}{G. Div\'eki and Cs. Imreh, Grid Based Online Algorithms for Clustering Problems, in Proc. 15th IEEE Int. Sympos. Comput. Intel. Infor. (CINTI), IEEE, pp 159, 2014}
	\bibitem{DKP}{S. Dobrev, R. Kr\'alovi\v{c}, E. Pardubsk\'a, Measuring the Problem-Relevant Information in Input, TAIRO - Theor. Inf. Appl., 43(3), pp 585--613, 2009}
	\bibitem{DT}{A. Dumitrescu, Cs. D. T\'oth, Online Unit Clustering in Higher Dimensions, Algorithmica 84, pp 1213--1231, 2022}
	\bibitem{EL}{M. R. Ehmsen, K. S. Larsen, Better Bounds on Online Unit Clustering, Theoretical Computer Science 500, 1--24, 2013}
	\bibitem{ELS}{L.Epstein, A. Levin, R. van Stee, Online Unit Clustering: Variations on a Theme, Theoretical Computer Science 407, pp 85--96, 2008} 
	\bibitem{ES}{L.Epstein, R. van Stee, On the Online Unit Clustering Problem, in: Proceedings of the 5th International Workshop on Approximation and Online Algorithms, pp 193--206, 2007}
	\bibitem{KK}{J. Kawahara, K. M. Kobayashi, An Improved Lower Bound for One-Dimensional Online Unit Clustering, Theoretical Computer Science 600, 171--173, 2015}
	\bibitem{Mik}{J. W. Mikkelsen, Randomization Can Be as Helpful as a Glimpse of the Future in Online Computation, In ICALP, LIPIcs 9, pp 1--14, 2016}
	\bibitem{ZZC}{H. Zarrabi-Zadeh, T. M. Chan, An Improved Algorithm for Online Unit Clustering. Algorithmica 54(4), 490--500, 2009}
\end{thebibliography}
\end{document}